\documentclass[manuscript]{acmart}
\usepackage{amsmath}
\usepackage{amsthm}
\newtheorem{proposition}{Proposition}
\usepackage[ruled,vlined]{algorithm2e}
\usepackage[frozencache,cachedir=.]{minted}
\usepackage{float}
\usepackage{booktabs}
\usepackage{bbm}
\newcommand{\MMD}{\mathrm{MMD}}
\AtBeginDocument{%
  \providecommand\BibTeX{{%
    \normalfont B\kern-0.5em{\scshape i\kern-0.25em b}\kern-0.8em\TeX}}}

\begin{document}

\title{Bandwidth-Optimal Random Shuffling for GPUs}

\author{Rory Mitchell}
\authornote{Both authors contributed equally to this research.}
\email{ramitchellnz@gmail.com}
\orcid{0000-0003-2892-1082}
\affiliation{%
  \institution{Nvidia}
  \streetaddress{2788 San Tomas Expressway}
  \city{Santa Clara}
  \state{CA}
  \country{USA}
  \postcode{95051}
}
\affiliation{%
  \institution{Waikato University}
  \streetaddress{Hillcrest Road}
  \city{Hamilton}
  \country{New Zealand}
  \postcode{3240}
}

\author{Daniel Stokes}
\authornotemark[1]
\affiliation{%
   \institution{nLIGHT, Inc}
  \streetaddress{4637 NW 18th Avenue}
  \city{Camas}
  \state{WA}
  \country{USA}
  \postcode{98607}
}
\affiliation{%
  \institution{Waikato University}
  \streetaddress{Hillcrest Road}
  \city{Hamilton}
  \country{New Zealand}
  \postcode{3240}
}

\author{Eibe Frank}
\author{Geoffrey Holmes}
\affiliation{%
  \institution{Waikato University}
  \streetaddress{Hillcrest Road}
  \city{Hamilton}
  \country{New Zealand}
  \postcode{3240}
}

\renewcommand{\shortauthors}{Mitchell and Stokes, et al.}

\begin{abstract}
Linear-time algorithms that are traditionally used to shuffle data on CPUs, such as the method of Fisher-Yates, are not well suited to implementation on GPUs due to inherent sequential dependencies, and existing parallel shuffling algorithms are unsuitable for GPU architectures because they incur a large number of read/write operations to high latency global memory. To address this, we provide a method of generating pseudo-random permutations in parallel by fusing suitable pseudo-random bijective functions with stream compaction operations. Our algorithm, termed `bijective shuffle' trades increased per-thread arithmetic operations for reduced global memory transactions. It is work-efficient, deterministic, and only requires a single global memory read and write per shuffle input, thus maximising use of global memory bandwidth. To empirically demonstrate the correctness of the algorithm, we develop a statistical test for the quality of pseudo-random permutations based on kernel space embeddings. Experimental results show that the bijective shuffle algorithm outperforms competing algorithms on GPUs, showing improvements of between one and two orders of magnitude and approaching peak device bandwidth.
\end{abstract}



\maketitle

\section{Introduction}
\label{sec:sample1}
Shuffling is a fundamental computer science problem, the objective of which is to rearrange a set of input elements into some pseudo-random order. The classical method of Fisher-Yates~\cite{fisher1943statistical} (popularised by Knuth in~\cite{Knuth:1997:ACP:270146}) randomly removes elements from an input buffer, one at a time, appending them to the output buffer. This algorithm runs optimally in $O(n)$ time, outputs uniformly distributed permutations, has a simple in-place variant, and is straightforward to implement. As such, sequential shuffling has long been considered a solved problem. However, modern computing hardware such as GPUs offer massively parallel computation that cannot be effectively exploited by the standard sequential shuffling approach due to dependencies inherent in Fisher-Yates type algorithms. The work presented in this paper was motivated by a gap in GPU-oriented parallel primitive libraries, such as Thrust~\cite{bell2012thrust}, Cub~\cite{merrill2015cub}, and Boost.Compute~\cite{boost}, which aim to implement the C++ standard library in equivalent form for GPUs, that prior to the method presented in this paper all lacked shuffling algorithms. The primary contributions of this paper are:
\begin{itemize}
  \item `Bijective shuffle': A new shuffling algorithm, carefully optimised for GPUs.
  \item A novel statistical test for shuffling algorithms, based on kernel space embeddings.
\end{itemize}

Some applications motivating this work are: permutation Monte Carlo tests~\cite{good2006permutation}, bootstrap sampling~\cite{mitchell2017accelerating}, approximation of Shapley values for feature attribution~\cite{mitchell2021sampling}, differentially private variational inference~\cite{prediger2021d3p}, and shuffle operators in GPU accelerated dataframe libraries~\cite{rapids}. More generally, whenever shuffling forms part of a high-performance GPU-based algorithm, performance is lost when data is copied back to the host system, shuffled using the CPU, and copied back to the GPU. Host-to-device copies are limited by PCIE bandwidth, which is typically an order of magnitude or more smaller than available device bandwidth~\cite{farber2011cuda}.

Informally, an ideal parallel algorithm would be capable of assigning elements of the input sequence randomly to unique output locations without contention, communicate minimally between processors, evenly distribute work between processors, use minimal working space, and have deterministic run-time. To achieve this goal, we utilise pseudo-random bijective functions, defining mappings between permutations. These bijective functions allow independent processors to write permuted elements to an output buffer, in parallel, without collision. Critically, we allow bijective functions defining permutations of larger sets to be applied to smaller shuffling inputs via a parallel compaction operation that preserves the pseudo-random property. The result is an $O(n)$ work algorithm for generating uniformly random permutations. We perform thorough experiments to validate the empirical performance and correctness of our algorithm. To this end, we also develop a novel test on distributions of permutations via unique kernel space embeddings.

We begin in Section \ref{sec:background} by introducing notation, summarising existing parallel shuffling algorithms, and describing GPUs. Section \ref{sec:bijective} introduces the idea of bijective functions, explaining their connection to shuffling and proposing two candidate bijective functions. Section \ref{sec:compaction} explains how compaction operations can be combined with bijective functions to generate bijections for arbitrary length sequences, forming the `bijective shuffle' algorithm. Section \ref{sec:statistical_tests} describes a statistical test comparing the kernel space embedding of a shuffling algorithm's output with its expected value. This test is used to evaluate the quality of the proposed shuffling algorithms and to select parameters. Finally, in Section \ref{sec:evaluation}, we evaluate the runtime and throughput of the proposed GPU shuffling algorithm.

\section{Background}
\label{sec:background}
Before discussing existing work, we briefly describe some notation. The order of elements in the shuffled version of an input sequence can be defined as a permutation. We refer to the symmetric group of permutations of $n$ elements as $\mathfrak{S}_n$. The permutation $\sigma \in \mathfrak{S}_n$ assigns rank $j$ to element $i$ by $\sigma(i)=j$. For example, given the permutation written in one-line notation:
$$\sigma = \begin{pmatrix}
  0 & 3 & 1 & 2
\end{pmatrix}$$
and the list of items 
$$(x_0,x_1,x_2,x_3)$$
the items are reordered such that $x_i$ occupies the $\sigma(i)$ coordinate
$$(x_0, x_2, x_3, x_1).$$

Note that we are using non-traditional zero-based indexing of permutations for notational convenience: this simplifies description of the functions used for our shuffling approach later in this paper.

The product of permutations $\sigma\tau$ refers to the composition operator, where the element $i$ is assigned rank $\sigma(\tau(i))=j$.

Shuffling is a reordering of the elements of an array of length $n$ by a random permutation on the finite symmetric group, $\sigma \in \mathfrak{S}_n$, such that each possible reordering is equally likely, i.e., $p(\sigma)=\frac{1}{n!}, \forall \sigma \in \mathfrak{S}_n$.

\subsection{Existing Work on Parallel Shuffling}
\label{sec:existing}
We analyse the computational complexity of shuffling algorithms in terms of the work-time framework~\cite{jeje1992introduction}, where an algorithm is described in terms of a number of rounds, where $p$ processors perform operations in parallel. GPUs are assumed to be \textit{data parallel} processors such that there is an independent processor for each data element. Algorithms are evaluated based on the notion of work complexity, or the total number of operations performed on all processors. As the Fisher-Yates shuffle runs in optimal $O(n)$ work on a single processor, a parallel algorithm is said to be \textit{work efficient} if it achieves $O(n)$ work complexity.

A parallel divide-and-conquer shuffling algorithm is independently proposed by Rao~\cite{rao} and Sandelius~\cite{sandelius} (we refer to this algorithm as RS). The input is recursively divided into subarrays by selecting a random number $0\cdots k-1$ until each subarray contains one element. The final permutation is obtained via an in-order traversal of all subarrays. Bacher et al.~\cite{mergeshuffle} provide a variant of this algorithm named \textit{MergeShuffle}, taking a bottom-up approach where the input is partitioned into $k$ subarrays that are randomly merged with their neighbours until the entire array is shuffled. The parallel work complexity of these recursive divide-and-conquer approaches is $O(n\log(n))$. The \textit{RS} and \textit{MergeShuffle} algorithms have strong parallels to integer sorting algorithms. In particular, \textit{RS} is equivalent to a most significant digit (MSD) radix sort~\cite{CoinFlippingAnalysis} on a sequence, where each element in the sequence is an infinite length random bit stream. Likewise, \textit{MergeShuffle} has strong parallels to merge sort, differing in the definition of the merge operator for each pass. We provide experiments in Section \ref{sec:evaluation} for divide-and-conquer style shuffling algorithms using GPU sorting.

Anderson~\cite{AndersonFisherYates} proves that each of the swaps in the Fisher-Yates algorithm can be reordered without biasing the generated permutation, assuming parallel processors implement an atomic swap operation and that atomic swaps are serialised fairly. Note that current GPU architectures do not serialise atomic swaps fairly, as no assumptions can be made as to the ordering of operations \cite{cuda}, so such an approach would be unsuitable for GPUs. Shun et al.~\cite{FisherYatesDependance} take a similar approach, proving that the Fisher-Yates algorithm has an execution dependence graph with a depth of $O(\log(n))$ with high probability, allowing parallel execution of the swaps, while generating the same output as the sequential approach.

The literature also has divide-and-conquer shuffling algorithms that focus on distributed environments. Sanders~\cite{Sanders} describes an algorithm for external memory and distributed environments where each item is allocated to a separate processor and permuted locally. A prefix sum is then used to place each element in the final output buffer. Langr et al.~\cite{MPIAlgorithm} provide a concrete extension of Sander's algorithm for use with the MPI library. Gustedt~\cite{Gustedt2003} also expands on this approach, describing a method of constructing a random communication matrix that ensures work is distributed fairly between all processors.

Reif~\cite{Reif} describes an algorithm that assigns each element a random integer key between 1 and $p$, where $p$ represents the number of processors. Elements are then assigned to processors by a sorting operation, where they are shuffled by the sequential Fisher-Yates method. The work complexity of the integer sort can be achieved in $O(\log(n))$ time using $n/\log(n)$ processors using the sort algorithm described by Reif~\cite{Reif}. A limitation of this approach is non-deterministic load balancing. One processor may receive significantly more work than others, limiting the algorithm to the speed of the slowest processor.

Alonso and Schott~\cite{AlonsoSchott} define a custom representation with an associated total ordering that can be `sorted' using a variant of merge-sort to produce a random permutation, falling into the class of divide-and-conquer algorithms. Their method forms a bijection between a so-called `lower-exceeding sequence' of length \textit{n} and a permutation of \textit{n} elements. It can be shown that there are exactly \textit{n!} such sequences of length \textit{n}. Given a uniformly selected random lower-exceeding sequence, the algorithm outputs the corresponding random permutation.

Czumaj et al.~\cite{Czumaj1998} make use of network simulation to define a shuffling algorithm. Their work describes two methods for randomly generating a network mapping input elements to a random output location. The first constructs a network representing a random Fisher-Yates shuffle. The network is processed to produce $n$ distinct keys. After sorting these keys, elements are efficiently mapped to output locations. All pre/post-processing steps require $O(\log(n))$ time using $n/\log(n)$ processors in linear space. The second method proposed by Czumaj et al.~\cite{Czumaj1998}, composes \textit{n}-way `splitters' to construct a network randomly permuting the input array. Each splitter randomly divides the input into two equal sized groups. For each splitter, the first group is ordered before the second in the final output. Each group is recursively split until the size of each group is one. At this stage, the network defines a random permutation, and is traversed to find the output position of each element. The algorithm runs in $O(c\log\log(n))$ time using $n^{1+1/c\log\log(n)}$ processors, for an arbitrary positive constant $c$. Granboulan and Pornin~\cite{BlockCipherCzumaj} use this method to generate a bijective function over an arbitrary input domain using $O(\log(n))$ space and $O(\log(n))$ time, but find the cost of the hyper-geometric random number generator used to construct splitters prohibitive.

`Dart-throwing' algorithms place input items randomly in an array of size $O(n)$ until each item is placed in a unique location. A compaction operation is applied to place items in a dense output array. When a placement collides with an occupied space, the process is retried until the element is placed successfully. Reif and Miller~\cite{DartThrowingTree} and later Reif and Rajasekaran~\cite{DartThrowingPrefix} describe a simple method in which a parallel prefix sum is used to perform the compaction. Matias and Vishkin~\cite{DartThrowingCanonical} describe a method using the canonical cycle representation to perform the compaction step, and Hagerup~\cite{DartThrowingMinPrefix} provides an alternative method for computing the min prefixes used when generating the cycle representation. Dart-throwing approaches achieve parallel work complexity of $O(n)$ in expectation, however, non-determinism makes them unsuitable for many practical applications. In our work, we similarly utilise prefix sum to perform compaction, but show how to obtain permutations in a completely deterministic way, using pseudo-random bijective functions in place of dart-throwing.

Cong and Bader~\cite{AlgoComparison} compare four divide-and-conquer, integer sorting, and dart-throwing algorithms and evaluate their performance relative to a sequential Fisher-Yates shuffle on up to 12 processors. They find that the sorting based approach is substantially worse than the other approaches, with Anderson's~\cite{AndersonFisherYates} approach tending to perform the best.

Closest to our work is the linear congruential generator (LCG) based approach of Andr{\'e}s and P{\'e}rez~\cite{LCGShuffle}, where permutations are generated by an LCG of full period. This approach is limited by the need to find viable LCG parameterisations for arbitrary length inputs, where the subset of available parameters is dramatically smaller than the space of permutations. In this paper we describe how any bijective function can be applied to the process of generating random permutations; in particular, we look at a more flexible use of LCGs and the use of \emph{n}-bit block ciphers with much stronger pseudo-random properties.

The above algorithms all suffer from one or more drawbacks with respect to an ideal GPU implementation. The divide-and-conquer approaches have sub-optimal $O(n\log n)$ work complexity. The Rao-Sandelius~\cite{rao} additionally suffers from load balancing issues, where the partitions generated by the divide-and-conquer process can be unevenly sized. The parallel Fisher-Yates algorithm of \cite{AndersonFisherYates} is non-deterministic and relies on fairly serialised atomics for correctness (GPU atomics are not fairly serialised). The network simulation algorithm of \cite{Czumaj1998} has a complicated implementation difficult to adapt to GPUs. Dart-throwing algorithms are non-deterministic both in terms of shuffle output and runtime. The simple LCG approach of \cite{LCGShuffle} offers no solution for arbitrary length sequences and suffers from poor quality pseudorandom outputs. In subsequent sections, we develop a new approach addressing these issues.

\subsection{Graphics Processing Units}
GPUs are massively parallel processors optimised for throughput, in contrast to conventional CPUs, which optimise for latency. GPUs in use today consist of processing units with single-instruction, multiple-thread (SIMT) lanes that efficiently execute instructions for groups of threads operating in lockstep. In the CUDA programming model, execution units called ``streaming multiprocessors'' (SMs), have 32 SIMT lanes. The corresponding group of 32 threads is called a ``warp''. Warps are generally executed on SMs without order guarantees, enabling latency of global memory loads to be hidden by switching between warps~\cite{cuda}.

Large speed-ups in the domain of GPU computing commonly occur when problems are expressed as a balanced set of vector operations with minimal control flow, and coalesced memory access patterns. Notable examples are matrix multiplication~\cite{matrix_multiplication,hall2003cache,tuning_matrix_multiplication}, image processing~\cite{efficient_dct,fft_gpu}, deep neural networks~\cite{Perry_2014,deep_learning,chetlur2014cudnn}, and sorting~\cite{merge_path,radix_sort}.

Implementation of shuffling algorithms poses a particular challenge for GPUs. The SIMD architecture favours a data-parallel approach, where work is evenly distributed among threads belonging to an execution unit --- if any one thread is slow to complete, the entire execution unit is stalled. This rules out approaches from Section \ref{sec:existing} with uneven work distribution per processor. Dart-throwing algorithms may be implemented for GPUs, using atomic compare and swap operators available in CUDA, although collisions for atomic compare and swap operations imply global synchronisation and are extremely costly. Of the existing approaches, sort-based approaches appear the most promising, given the availability of state-of-the-art GPU sorting. Highly optimised algorithms exist for both merge-sort and LSB radix sort via the Thrust~\cite{bell2012thrust} library. As mentioned in Section \ref{sec:existing}, the problem of shuffling can be expressed as a sort over a list whose elements are infinite length random bit strings. In fact, the keys do not need to be infinite, only long enough to break any ties in comparisons between elements. One baseline GPU algorithm we consider is to generate random integer sort keys of machine word length (assumed to be 64 bits) and perform a key-value sort. GPU sorting algorithms such as merge-sort and LSB radix sort perform several passes ($O(\log n)$ for merge-sort and $O(k)$ for radix sort, with $k$ proportional to the sort key size in bits) scattering elements in memory. These scatter passes are particularly expensive for random keys, as memory writes cannot be coalesced together efficiently, and represent the largest performance bottleneck for this shuffling algorithm.

While the above algorithms are candidates for implementation on GPUs, we may improve shuffle throughput and device utilisation significantly by devising a new algorithm tailored to the architecture. As the shuffle operation must at minimum reorder elements using a gather or scatter operation, we posit that the maximum bandwidth of any random shuffle algorithm implemented for GPUs is that of a random gather or scatter. Figure \ref{fig:gather} demonstrates a parallel random gather operation, where threads read from noncontiguous memory locations in the input buffer and write to contiguous memory locations in the output buffer. A scatter operation is the inverse, where threads read from contiguous memory locations and write to noncontiguous memory locations. Scatter/gather operations for GPUs are discussed in detail in \cite{he2007efficient}. As our experiments show that gather operations have a higher bandwidth than scatter operations, we focus on the former. 
\begin{figure}
    \centering
    \includegraphics[width=0.75\columnwidth, keepaspectratio]{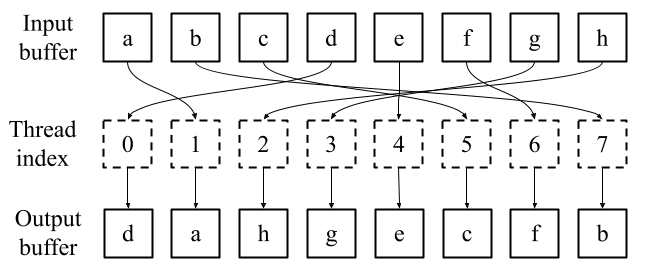}
    \caption{Parallel random gather --- Threads read from non-contiguous memory locations and write to contiguous memory locations.}
    \label{fig:gather}
\end{figure}

Having established an upper bound on GPU bandwidth, we develop an optimised shuffling algorithm operating in a single gather pass, to approach the theoretical peak performance of the device.

\section{Shuffling With Bijective Functions}
\label{sec:bijective}

Our proposed approach is based on applying bijective functions. The symmetric group $\mathfrak{S}_n$, defined over any set of $n$ distinct items, contains all bijections of the set onto itself, or all distinct permutations of $n$ elements. A bijection is a one-to-one map between the elements of two sets, where each element from the first set is uniquely paired with an element from the second set. A bijection between two sets $X,Y$ with $|X|,|Y|=n$ is characterised by some function $f_n: X \rightarrow Y$. The bijective function $f_n$ admits a corresponding inverse $f_n^{-1}: Y \rightarrow X$, which is also a bijection. For any two such functions $(f_n,g_n)$, their composition $f_n \circ g_n$ is also a bijection. As we are dealing with bijections from the symmetric group $\mathfrak{S}_n$, we define functions of the form $f_n: X \rightarrow X$. Without loss of generality, consider the set of nonnegative integers $X=\{0,1,\cdots,n-1\}$. An example of a bijection $\sigma \in \mathfrak{S}_4$ is
\begin{align*}
f_4(0) = 2\\
f_4(1) = 3\\
f_4(2) = 1\\
f_4(3) = 0
\end{align*}

If each of the $n$ functions can be executed on a processor $p_i$, independently of any other processor, without communication and in reasonable time, then this defines a simple parallel algorithm for reordering elements. If $f_n$ furthermore exhibit suitable pseudo-random properties such that $p(\sigma) = \frac{1}{n!}, \forall \sigma \in \mathfrak{S}_n$ (or close enough for practical application), then we have an effective parallel shuffling algorithm. We now discuss potential candidate bijective functions.

\subsection{Linear Congruential}
A first candidate is based on the common linear congruential random number generator (LCG). Given constants $a,c,$ and $n$, the LCG outputs
\begin{equation}
\label{eq:lcg}
y=(a x)+c \text{ mod } n
\end{equation}
If $a$ and $n$ are co-prime, it is well-known that the input $x \in {0,1,\cdots, n-1}$ maps to a unique output location in the integer ring $\mathbb{Z} / n \mathbb{Z}$. Therefore, Equation \ref{eq:lcg} defines a bijective function and may be used to create permutations over inputs of length $n$. For now, assume $n$ is fixed at the length of the input sequence. To find a bijection for a fixed $n$, the task reduces to finding some $a$ co-prime to $n$ ($c$ is unrestricted). Finding $a$ is trivial if $n$ has certain properties, for example, any $a<n$ where $n$ is prime, or odd $a<n$ where $n$ is a power of two. We will later show how to modify the length of the input sequence such that co-prime $a$ is always easily available.

The method described above is simple to implement and computationally inexpensive; however, linear congruential generators are known to have weaknesses as random number generators~\cite{Press2007}. For our use case in particular, assuming some $n$ fixed relative to input length and varying $a$ and $c$ from the above, we can achieve at most $n^2$ unique permutations --- significantly less than the $n!$ possible permutations. Hence, permutations generated by this method may be appropriate for basic applications, but a more robust permutation generator is desirable. 

\subsection{Feistel Network}
\label{sec:feistel}
A second candidate bijective function is a block cipher construction known as a Feistel network~\cite{feistel1973}. A cryptographic block cipher, accepting an encryption key and a $b$-bit plain-text block, provides a bijective mapping to a $b$-bit cipher-text block. A perfect block cipher outputs cipher-text computationally indistinguishable from a random bit string, when the key is unknown. Thus, for a random choice of key, the output should be computationally indistinguishable from a random permutation. The Feistel network construction is a core component of many modern encryption algorithms, such as the Data Encryption Standard (DES) and so represents a significant improvement over the LCG in terms of the quality of its pseudo-random outputs~\cite{Biryukov2005}. 

\begin{figure}
    \centering
    \includegraphics[width=0.75\columnwidth, keepaspectratio]{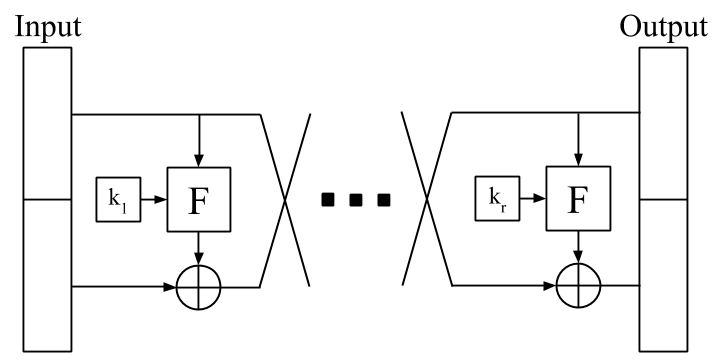}
    \caption{Feistel Network construction --- Successive rounds of the function $F $ are applied to the input, each time with a unique key.}
    \label{fig:FeistelNetwork}
\end{figure}

The construction for a $b$-bit Feistel network consists of multiple applications of a round function as shown in Figure \ref{fig:FeistelNetwork}. In round $i$, the input is split into a $\lfloor b/2\rfloor$-bit binary string $L$, representing the left half of the input block and a $\lceil b/2\rceil$-bit binary string $R$ representing the right half of the input block. A round function $F$ is applied to $R$ and a round key $k_i$. $F$ does not need to be bijective. The first $\lfloor b/2\rfloor$-bit output from $F$ is combined with $L$ using the exclusive-or operation to give the right half for round $i+1$. The right half from round $i$ becomes the left half for round $i+1$. Thus, the $i$th round of the Feistel Network is defined as

\begin{equation}
    f_i(L|R)=R|(L\oplus \ F(R,k_i)),
\end{equation}

with $\oplus$ the bitwise exclusive-or operation. Consequently, a block cipher constructed using a Feistel Network construction with $r$ rounds can be defined as

\begin{equation}
    g(L|R)=f_r\circ\cdots f_2\circ f_1(L|R)
\end{equation}

Luby and Rackoff~\cite{LubyRackoffPRP} proved that a three-round Feistel network is a pseudo-random permutation if the round function $F$ is a pseudo-random function. Therefore, a Feistel Network with $rounds\ge3$ can be used to generate a pseudo-random permutation for any set of size $n=2^b$. 

Feistel networks have been shown to be effective for parallel random number generation. Salmon \textit{et al}.~\cite{salmon} describe two hardware-efficient pseudo-random number generators (PRNGs), named Philox and Threefry, based on simplifications of cryptographic block ciphers. These PRNGs are sometimes called \textit{counter-based}, as generation of the $i$'th random number $x_i$ is stateless, requiring only the index $i$ and a random seed. Advancing the series does not require $x_{i-1}$, and so may be performed trivially by parallel threads. This is in contrast to the canonical Mersenne twister~\cite{mersenne}, which requires a state of 2.5kB and cannot advance the series arbitrarily from $x_i$ to $x_{i+m}$ in constant time with respect to $m$. The Philox and Threefry generators are shown to produce at least $2^{64}$ independent streams of random numbers, with a period of $2^{128}$ or more, and pass BigCrush~\cite{l2007testu01} statistical tests. Figure \ref{fig:philox} shows one step of the Philox PRNG, which differs from the standard Feistel network in the addition of the function $B_k$, which is strictly a bijection. Philox makes use of fast integer multiplication instructions, where multiplication by a carefully chosen, odd constant, yields upper bits forming $F$ and lower bits modulo $2^w$ form the bijection $B$. $B$ is guaranteed to be a bijection because the odd multiplicand is always coprime to the integer ring modulo $2^w$. 

Although these PRNGs happen to have the bijective property, the connection between bijections and the symmetric group $\mathfrak{S}_n$ is not exploited in \cite{salmon}. In this paper, we adapt the Philox cipher to generate bijections with lengths of arbitrary powers of two instead of $2^{64}$, while preserving the invertibility of the process, terming our modification \textit{VariablePhilox}. For example, shuffling a sequence of size $2^7$ results in halves with sizes $|L|=3,|R|=4$ and the bijective property is lost for the standard Philox cipher. Figure \ref{fig:variable_philox} shows a modified construction, where $b$ is the odd bit and $|L|=|R|$. $F_k$ is a pseudo-random key-dependent function, and $B_k$ is a key-dependent bijective function. $G$ is a function mixing bit $b$ into $R'$, where for $G(B_k(L),b)=(R',b')$, $b$ is inserted in the least significant position of $B_k(L)$, yielding $R'$. Then the most significant bit of $R'$ is removed and returned as $b'$. $G$ is clearly invertible such that $G^{-1}(R', b')=(B_k(L),b)$. Defining $\oplus$ as bitwise exclusive-or, \textit{VariablePhilox} encodes inputs as
\begin{align*}
L'= F_k(L) \oplus R \\
(R',b') = G(B_k(L),b)
\end{align*}
This process is inverted to retrieve $(L,R,b)$ by computing (in order):
\begin{align*}
(B_k(L),b)=G^{-1}(R',b') \\
L=B_k^{-1}(B_k(L))\\ 
R=F_k(L) \oplus L'
\end{align*}
Invertibility guarantees the bijective property, so \textit{VariablePhilox} is a bijection for arbitrary power-of-two length sequences. C++ reference code is given in Listing \ref{lst:philox}. 
\begin{figure}
\centering
\begin{minipage}[b]{.366\linewidth}
    \centering
    \includegraphics[width=\textwidth]{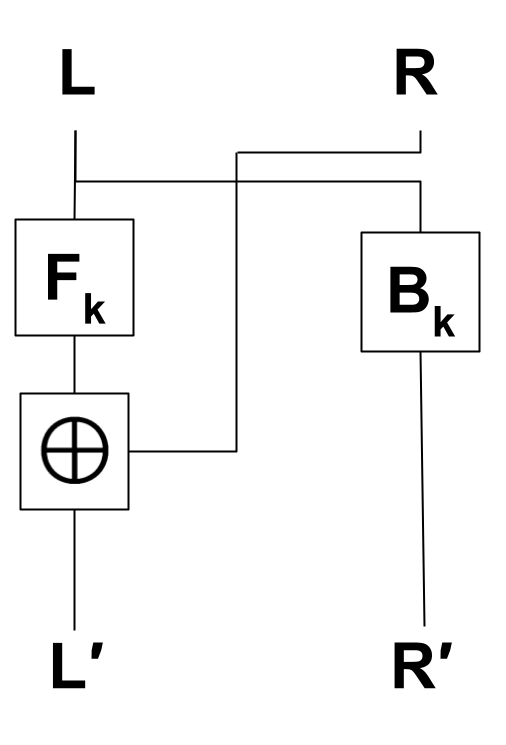}
    \caption{Philox cipher --- Bijective function $B_k$ is added to the standard Feistel cipher}
    \label{fig:philox}
\end{minipage}
\hfill
\begin{minipage}[b]{.534\linewidth}
    \centering
    \includegraphics[width=\textwidth]{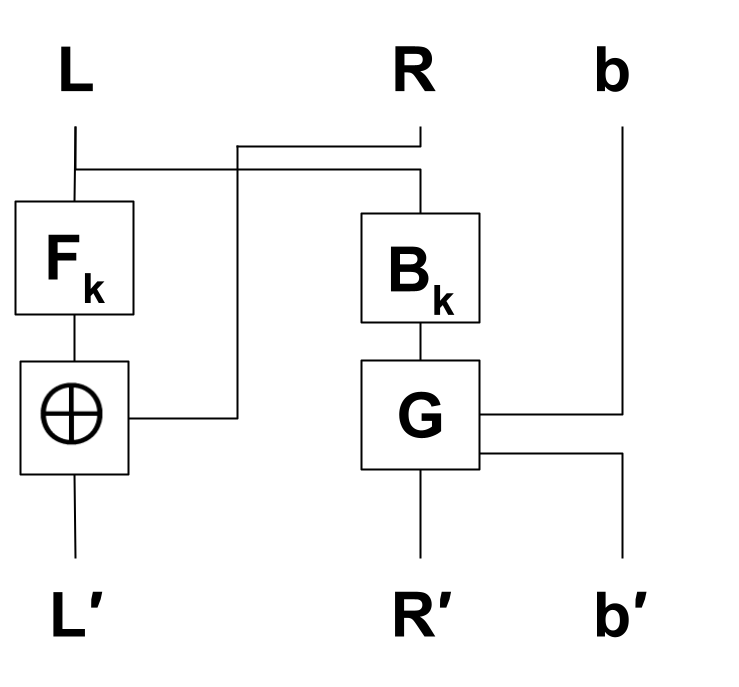}
    \caption{VariablePhilox cipher --- the odd bit $b$ (if it exists) is mixed into $R'$ by function $G$.}
    \label{fig:variable_philox}
\end{minipage}
\end{figure}

We have now described several methods of generating bijections with pseudo-random properties, but these bijections do not apply to arbitrary input lengths. In particular, the Feistel bijective functions are available for sequences of power of two length. We now show how to efficiently extend these to arbitrarily sequences.

\begin{listing}
\begin{minted}
{C++}
uint64_t VariablePhilox(const uint64_t val) const
{
  static const uint64_t M0 = UINT64_C(0xD2B74407B1CE6E93);
  uint32_t state[2] = { uint32_t(val >> right_side_bits),
    uint32_t(val & right_side_mask)
  };
  for (int i = 0; i < num_rounds; i++)
  {
    uint32_t hi;
    uint32_t lo = mulhilo(M0, state[0], hi);
    lo = (lo << (right_side_bits - left_side_bits)) |
      state[1] >> left_side_bits;
    state[0] = ((hi ^ key[i]) ^ state[1]) & left_side_mask;
    state[1] = lo & right_side_mask;
  }
  // Combine the left and right sides together to get result
  return (uint64_t) state[0] << right_side_bits |
    (uint64_t) state[1];
}
\end{minted}
\caption{VariablePhilox implementation --- M0 is a constant selected in \cite{salmon}, mulhilo performs 64 bit integer multiplication, returning the result as the upper and lower 32 bits.}
\label{lst:philox}
\end{listing}

\section{Arbitrary Length Bijections}
\label{sec:compaction}
The key observation for generating bijections of arbitrary length is this: given an input vector $X=[0,1,2,\cdots, m-1]$ of length $m$, there may not be a readily available bijective function $f_m$ of the same length, however, we can find the nearest applicable $f_n$ such that $m \leq n$ and apply the bijective function to the padded vector $\hat{X}=[0,1,2,\cdots, n-1]$. The output is a vector $W$ of length $n$, containing the randomly permuted input indices. By `deleting' all $w \geq m$ from $W$, we obtain a permutation $Y$ of length $m$. Using this process, we form an algorithm generating longer permutations of length $n \geq m$, then compacting them to form permutations of length $m$. This compaction of a longer random permutation is still an unbiased random permutation as per the following proposition.

\begin{proposition}
Define the function $t: \mathfrak{S}_n \rightarrow \mathfrak{S}_m$ removing elements of $\sigma \in \mathfrak{S}_n$ where $\sigma(i) > m$, returning the permutation $\tau \in \mathfrak{S}_m$ of length $m$. If $p(\sigma)=\frac{1}{n!}, \forall \sigma \in \mathfrak{S}_n$, then $p(\tau)=\frac{1}{m!}$.
\end{proposition}
\begin{proof}
The function $t(\sigma)=\tau$ is a surjective map from $\mathfrak{S}_m \rightarrow \mathfrak{S}_n$. For each $\tau \in \mathfrak{S}_n$, there exists a non-overlapping subset $\Pi_{\tau} \subseteq \mathfrak{S}_m$ such that $\forall \pi \in \Pi_{\tau},\: t(\pi)=\tau$. Clearly $|\Pi_{\tau}|=\frac{m!}{n!}, \forall \tau \in \mathfrak{S}_n$ and so, if $p(\sigma)=\frac{1}{m!}, \forall \sigma \in \mathfrak{S}_m$, then $p(\tau)=\frac{1}{n!}$.
\end{proof}

\begin{figure}
    \centering
    \includegraphics[width=0.75\columnwidth]{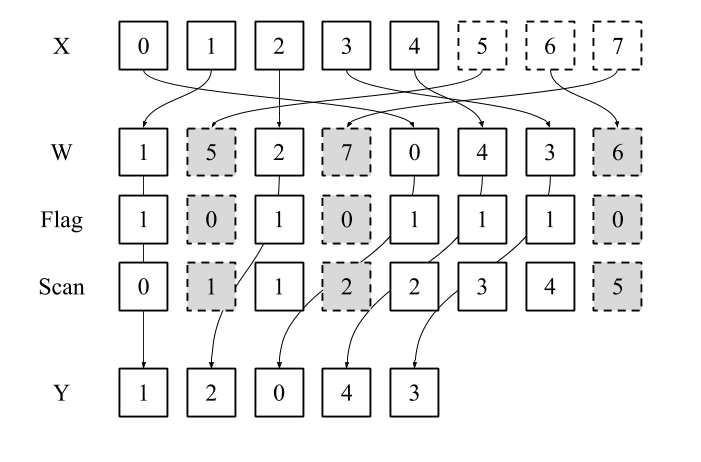}
    \caption{Shuffle compaction --- Dashed boxes represent dummy elements used to extend the sequence to a power of two. The bijective function is evaluated for each thread, yielding $W$.  We flag any $W > 4$, then the scan of these flags provides the destination address for each $X[W_i]$, where $W_i \leq 4$.}
    \label{fig:shuffle}
\end{figure}

\begin{algorithm}
\SetKwInOut{Input}{input}
        \SetKwInOut{Output}{output}
 \Input{$m,n,X,f_n$}
 \Output{$Y$}
 \For{$i=0$ \KwTo $n - 1$ \bf{in parallel}}{
  \tcp{Evaluate bijective function}
  $b = f_n(i)$\;
  $flag = 1$\;
  \If{$b \geq m$}{
   $flag = 0$\;
   }
   \tcp{Find output location of valid elements}
   $out\_idx= exclusive\_scan(flag)$\;
   \tcp{Gather elements to output}
   \If{$b < m$}{
   $Y[out\_idx] = X[b]$\;
   }
 }
 \caption{Bijection Shuffle}
\label{alg:shuffle}
\end{algorithm}

To achieve the compaction in parallel, we use the work-efficient exclusive scan of Blelloch~\cite{blelloch}. Flagging each element of the extended permutation $W$ with $0$ if $w_i \geq m$ or $1$ if $w_i < m$, the output of the exclusive scan over these flags gives the scatter indices into the shuffled output vector $Y$. This process is shown in Figure \ref{fig:shuffle} for $m=5,n=8$, and pseudocode is given in Algorithm \ref{alg:shuffle}. Given $p=n$ independent threads, the exclusive scan operation has work complexity $O(n)$. Our bijective function, executed in parallel for each element, also has work complexity $O(n)$, and so the final algorithm has optimal work complexity $O(n)$. Note that, using the bijective functions described in Section \ref{sec:bijective}, the padded array of length $n$ will have at most $2m$ elements (the nearest power of two length), so the work complexity holds regardless of input length.

Algorithm \ref{alg:shuffle} can be implemented easily in several steps using a GPU parallel primitives library, for example, evaluating the bijective function, gathering elements to an output buffer, then applying a stream compaction algorithm. Achieving ``bandwidth optimality'', i.e., a single global memory read and write per element, is more challenging, requiring the fusion of all operations into a single GPU kernel. The standard GPU parallel prefix sum of \cite{harris2007parallel} uses two passes through global memory, the first to perform block-level scans and the second to propagate partial block-level results globally. This can be improved upon by using the specialised single-pass scan implementation of \cite{merrill2016single}, using a technique termed \textit{`decoupled look-back'} to achieve non-blocking communication between thread blocks, achieving optimal $n$ global memory reads/writes. In Section \ref{sec:evaluation}, we evaluate three versions of our GPU shuffling algorithm, incorporating varying levels of kernel fusion and demonstrating the effectiveness of optimising global memory read/writes for the shuffling problem.

\section{Statistical Tests}
\label{sec:statistical_tests}

We now consider statistical tests for pseudo-random permutations to assess the suitability of the methods discussed in Section \ref{sec:bijective} and \ref{sec:compaction} for practical applications, and to select the number of rounds for the \textit{VariablePhilox} cipher. Testing the distribution of permutations is challenging as the sampling space grows super-exponentially with the length of the permutation. In particular, $21! > 2^{64}$, so any computer algorithm for shuffling that uses a 64 bit integer seed cannot generate all permutations for $n \geq 21$.

Our first test considers the distribution of random permutations at $n=5$, where $5!=120$, using a standard $\chi^2$ test with 119 degrees of freedom, under the null hypothesis that permutations are uniformly distributed, i.e., each permutation occurs with probability $p=\frac{1}{120}$. Figure \ref{fig:chi_squared} shows the change in the $\chi^2$ statistic as the number of rounds in the \textit{VariablePhilox} cipher is increased. Each data point is computed from 100,000 random permutations. We also plot the acceptance thresholds for $\alpha=0.05$, corresponding to the probability of observing a value of the $\chi^2$ statistic above the line indicated on the figure, if the null hypothesis is true. Considering the results shown in the figure, it is unlikely that samples from \textit{VariablePhilox} with less than 20 rounds are drawn from a uniform distribution. This is somewhat surprising given that only 10 rounds are recommended for the original Philox cipher in the PRNG setting~\cite{salmon}. The LCG bijective function generates a $\chi^2$ statistic in the region of 500,000, so it clearly fails the test and is not included in this figure.

The $\chi^2$ test is useful for small $n$, but is intractable otherwise. We develop another test statistic suited to larger $n$, based on the maximum mean discrepancy (MMD) in a reproducing kernel Hilbert space (RKHS). Two-sample hypothesis tests using MMD are developed in \cite{two_sample}. In a similar spirit, we derive a one sample test comparing the uniform distribution of permutations against a finite sample generated by a shuffling algorithm. The $\MMD^2$ between two distributions $p(X)$ and $q(Y)$ in a RKHS $\mathcal{H}$, equipped with positive definite kernel $K$, is defined as
\begin{align}
\label{eq:mmd}
\MMD^2(p,q) = \mathbb{E}_{x,x'}[K(x,x')] - 2\mathbb{E}_{x,y}[K(x,y)] + \mathbb{E}_{y,y'}[K(y,y')].
\end{align}

If the kernel $K$ is a \textit{characteristic kernel}, the mean embedding of a distribution induced by the kernel is injective~\cite{injective}. In other words, the mean embedding of a distribution is unique to that distribution. As a consequence, $\MMD(p,q)=0$ if and only if $p=q$. Thus, a strategy for statistical testing is to form the null hypothesis that $p=q$, compute the sample estimate $\hat{\MMD^2}(p,q)$ as the test statistic, and evaluate the probability of obtaining a sample estimate greater than some threshold, assuming $p=q$, using a concentration inequality. If the observed test statistic is sufficiently unlikely, this provides evidence that $p\neq q$.

To implement this idea, a characteristic kernel measuring the similarity of two permutations is needed. The Mallows kernel, for $\lambda \geq 0$, is defined for permutations as
$$K^{\lambda}_M(\sigma, \sigma') = e^{-\lambda n_{\textrm{dis}}(\sigma, \sigma')/\binom{n}{2}}.$$
where 
$$n_{\textrm{dis}}(\sigma, \sigma') = \sum_{i<j} [\mathbbm{1}_{\sigma(i) < \sigma(j)}\mathbbm{1}_{\sigma'(i) > \sigma'(j)} + \mathbbm{1}_{\sigma(i) > \sigma(j)}\mathbbm{1}_{\sigma'(i) < \sigma'(j)}].$$
We (somewhat arbitrarily) use the parameter $\lambda=5$ throughout this paper. The Mallows kernel is introduced in \cite{permutation_kernels} and shown to be characteristic in \cite{mania2018kernel}. It may be implemented in time $O(n \log n)$ using the procedure presented in \cite{kt_complexity}. In the following, we make use of the fact \cite{mitchell2021sampling} that the expected value of the Mallows kernel under a uniform distribution of permutations is
$$\forall \sigma \in \mathfrak{S}_n, \quad \mathbb{E}_{\sigma'}[K(\sigma,\sigma')]=\prod_{j=1}^{n}\frac{1-e^{-\lambda j / \binom{n}{2}}}{j(1-e^{-\lambda / \binom{n}{2}})}.$$
The Mallows kernel is right-invariant in the sense that $K(\sigma, \sigma')= K(\tau\sigma, \tau\sigma')$ for any $\tau\in  \mathfrak{S}_d$~\cite{diaconis1988group}. Also note that the uniform distribution of permutations is invariant to composition. Using right invariance in conjunction with $\tau=\sigma^{-1}$, and assuming $p$ denotes the uniform distribution over permutations, then 
\begin{align*}
\mathbb{E}_{\sigma' \sim p}[K(\sigma,\sigma')] &= \mathbb{E}_{\sigma' \sim p}[K(\tau\sigma,\tau\sigma')] \\
    &= \mathbb{E}_{\sigma' \sim p}[K(I,\tau \sigma')] \\
    &= \mathbb{E}_{\sigma' \sim p}[K(I,\sigma')]
\end{align*}
with $I$ the identity permutation. Using the above, and assuming that $p$ is a uniform distribution of permutations, the $\MMD^2$ from \eqref{eq:mmd} is simplified as
\begin{align*}
\MMD^2(p,q) &= \mathbb{E}_{x}[K(I,x)] - 2\mathbb{E}_{x}[K(I,x)] + \mathbb{E}_{y,y'}[K(y,y')]  \\
&= \mathbb{E}_{y,y'}[K(y,y')] - \mathbb{E}_{x}[K(I,x)] \\
&= \mathbb{E}_{y,y'}[K(y,y')] - \prod_{j=1}^{n}\frac{1-e^{-\lambda j / \binom{n}{2}}}{j(1-e^{-\lambda / \binom{n}{2}})}.
\end{align*}
Replacing $q$ with a sample of random permutations $\Pi$, of size $|\Pi|$ and with $|\Pi| \equiv 0 \pmod 2$, we obtain an unbiased sample estimate via
\begin{align}
\label{eq:mmd_estimate}
\hat{\MMD}^2(p,\Pi) &= \frac{2}{|\Pi|} \sum_{i}^{|\Pi|/2}K(\Pi_{2i-1},\Pi_{2i}) - \prod_{j=1}^{n}\frac{1-e^{-\lambda j / \binom{n}{2}}}{j(1-e^{-\lambda / \binom{n}{2}})},
\end{align}
where $\mathbb{E}[\hat{\MMD}^2(p,\Pi)]=0$ if and only if $p=q$. We derive two tests based on the $\hat{\MMD}^2(p,\Pi)$ statistic. The first is a distribution-free bound. Applying Hoeffding's inequality~\cite{inequalities}, with the fact that $0 \leq K(\sigma,\sigma') \leq 1$, we have
\begin{align*}
P(|\hat{\MMD}^2(p,\Pi)| \geq t)  \leq 2 \exp(-|\Pi| t^2)
\end{align*}
Let $\alpha_H$ be the level of significance under the Hoeffding bound. With null hypothesis $p=q$, the acceptance region of the test is
\begin{align}
|\hat{\MMD}^2(p,\Pi)| < \sqrt{\frac{\log(2/\alpha_H)}{|\Pi|}}.
\end{align}

The second test uses the central limit theorem, implying that $\hat{\MMD}^2(p,\Pi)$ approaches a normal distribution as $|\Pi| \rightarrow \infty$. The variance of the Mallows kernel may be computed directly from its expectation as
\begin{align*}
    \mathrm{Var}(K(\sigma,\sigma')) &= \mathbb{E}[K(\sigma,\sigma')^2] - \mathbb{E}[K(\sigma,\sigma')]^2 \\
    &= \mathbb{E}[e^{-2\lambda n_{\mathrm{dis}}(\sigma,\sigma')/\binom{n}{2}}]-\mathbb{E}[e^{-\lambda n_{\mathrm{dis}}(\sigma,\sigma')/\binom{n}{2}}]^2 \\
    &= \prod_{j=1}^{n}\frac{1-e^{-2\lambda j / \binom{n}{2}}}{j(1-e^{-2\lambda / \binom{n}{2}})} - \left( \prod_{j=1}^{n}\frac{1-e^{-\lambda j / \binom{n}{2}}}{j(1-e^{-\lambda / \binom{n}{2}})} \right)^2.
\end{align*}
So we have
$$\mathrm{Var}(\hat{\MMD}^2(p,\Pi)) = \frac{2 \cdot \mathrm{Var}(K(\sigma,\sigma'))}{|\Pi|}.$$
The factor of two arises because the sum in \eqref{eq:mmd_estimate} is of size $|\Pi|/2$. According to a normal distribution with mean 0 and variance as above,
$$P(|\hat{\MMD}^2(p,\Pi)| \geq t) = 1 - \mathrm{erf}\left(\frac{t}{\sqrt{2 \mathrm{Var}(\hat{\MMD}^2)}}\right),$$
where $\mathrm{erf}$ is the canonical error function.
The acceptance region for $\alpha_{N}$ is
\begin{align}
|\hat{\MMD}^2(p,\Pi)| < \sqrt{2 \mathrm{Var}(\hat{\MMD}^2)} \, \mathrm{erf}^{-1}(1-\alpha_N).
\end{align}

The threshold for $\alpha_H$ should be used for small $|\Pi|$ (i.e. $<100$), as it makes no assumptions on the distribution of $\hat{\MMD}^2(p,\Pi)$. For larger sample sizes, the asymptotic $\alpha_N$ threshold is significantly tighter and should be preferred. Comparing Figures \ref{fig:chi_squared} and \ref{fig:mmd_n5}, we see the MMD tests with the asymptotic acceptance threshold roughly coincide with the $\chi^2$ test at n = 5, where some tests fail for VariablePhilox with fewer than 20 rounds, but VariablePhilox with 24 rounds or more passes all tests. Figures \ref{fig:mmd_n5}, \ref{fig:mmd_n100}, and \ref{fig:mmd_n1000} plot the $|\hat{\MMD}^2(p,\Pi)|$ statistic for VariablePhilox, LCG, and std::shuffle using $|\Pi|=$100,000, for varying permutation lengths. These experiments lead us to recommend a 24 round VariablePhilox cipher for random permutation generation, and we use this configuration in all subsequent experiments.

\begin{figure}
\centering
\begin{minipage}[b]{.49\linewidth}
    \centering
    \includegraphics[width=1\columnwidth]{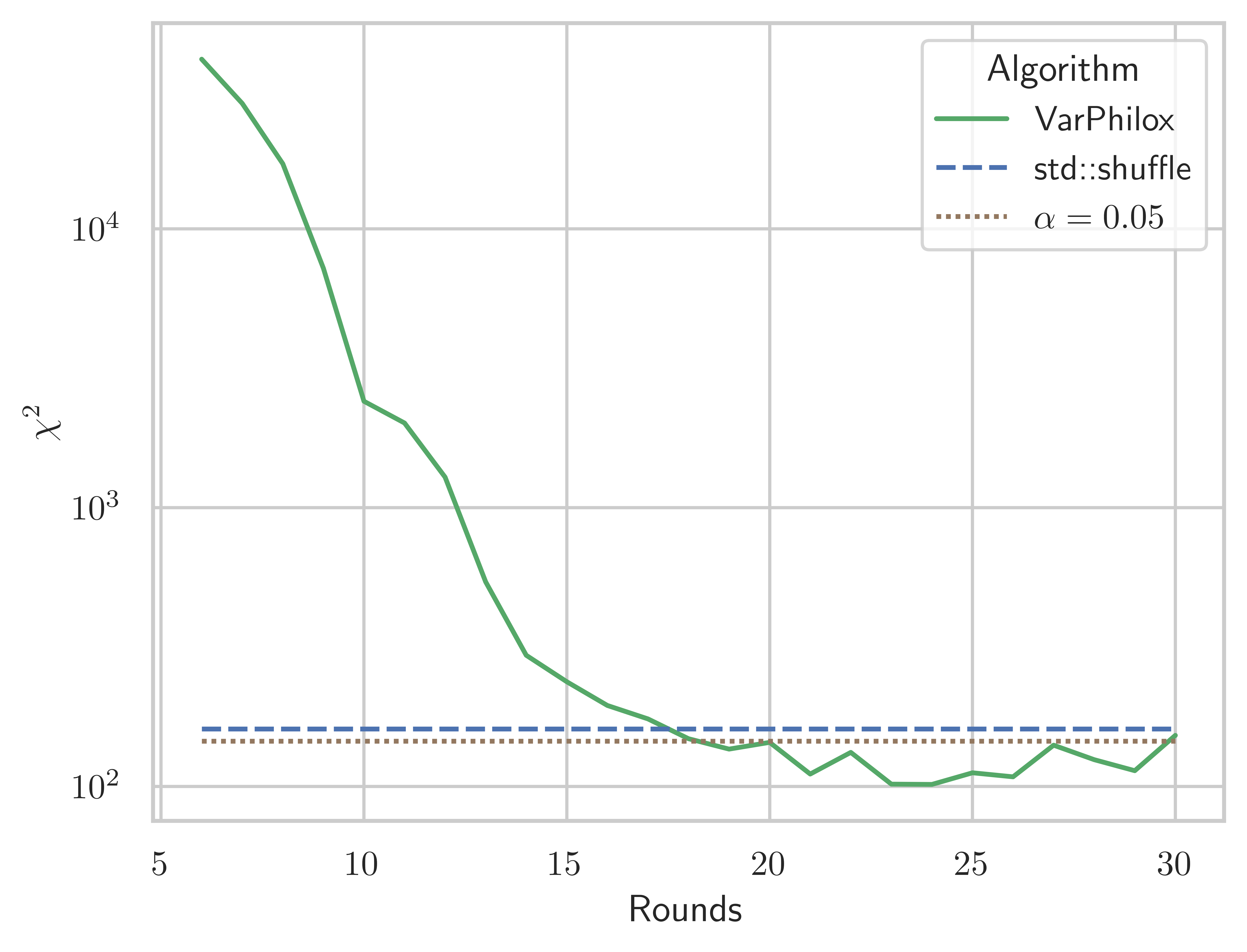}
    \caption{$\chi^2$ statistic vs. rounds, $|\Pi|=$100,000}
    \label{fig:chi_squared}
\end{minipage}
\begin{minipage}[b]{.49\linewidth}
    \centering
    \includegraphics[width=1\columnwidth]{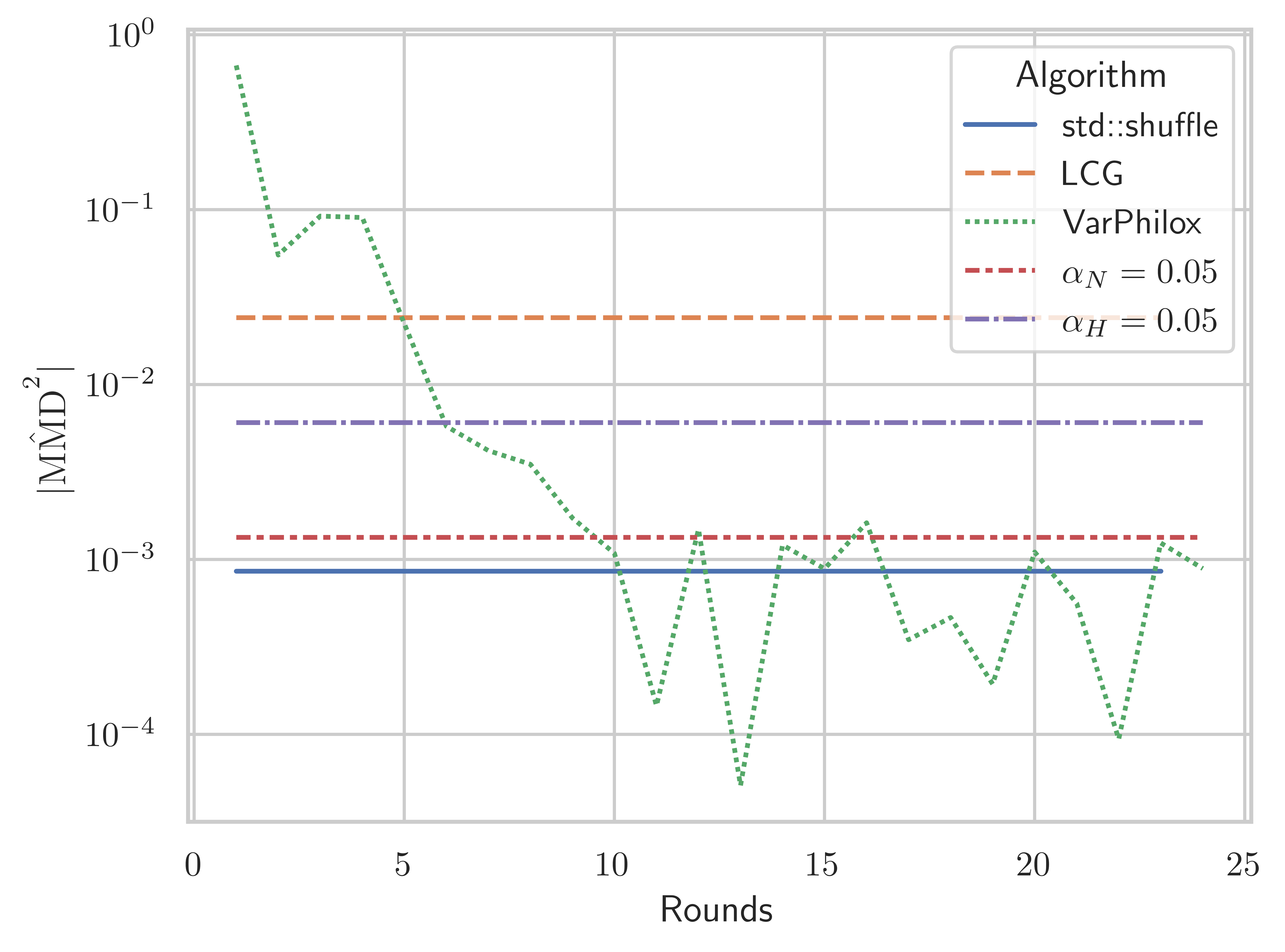}
    \caption{$|\hat{\MMD}^2(p,\Pi)|$ statistic, $n=5$, $|\Pi|=$100,000}
    \label{fig:mmd_n5}
\end{minipage}
\\
\vspace{0.6cm}
\begin{minipage}[b]{.49\linewidth}
    \centering
    \includegraphics[width=1\columnwidth]{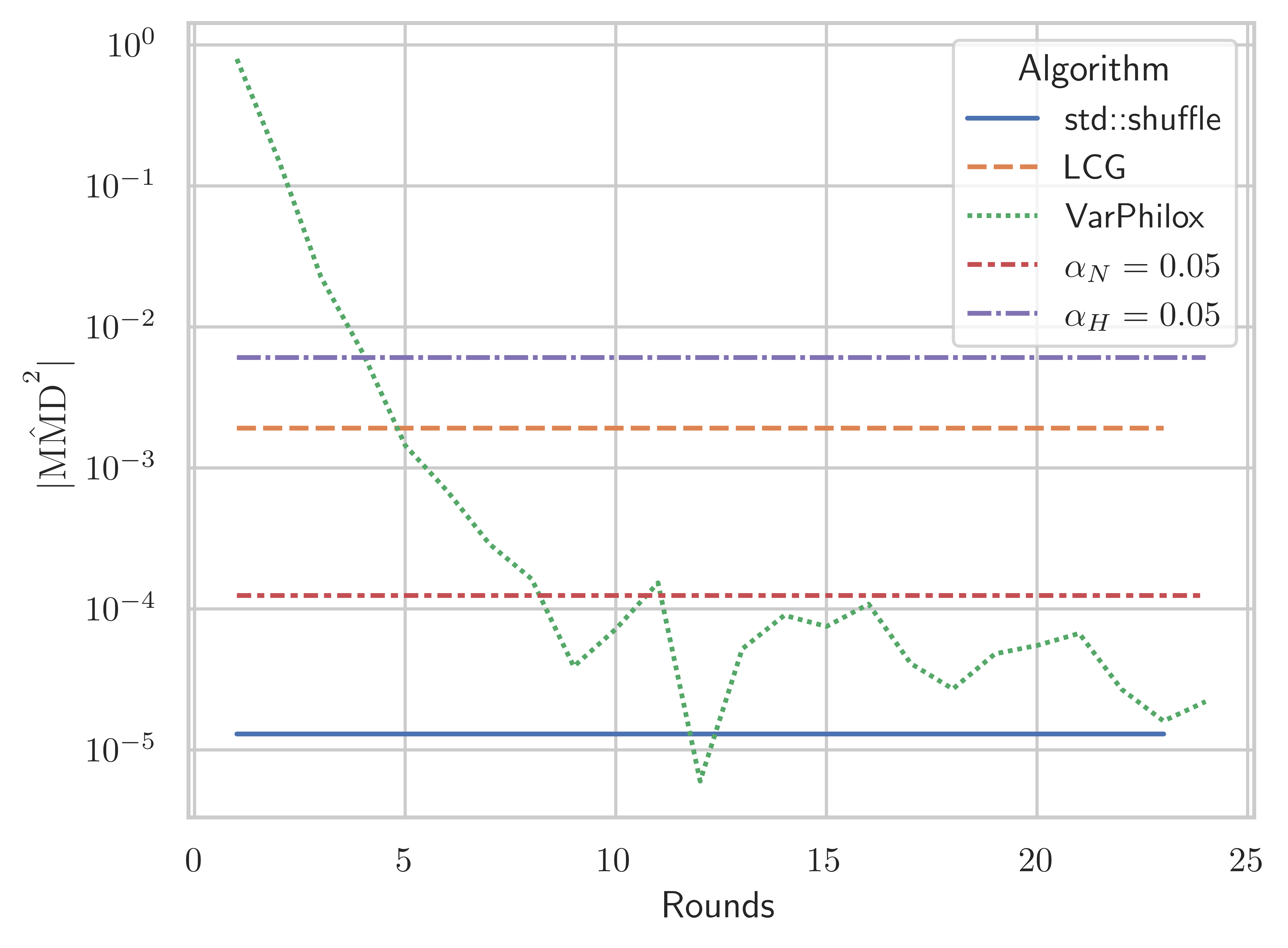}
    \caption{$|\hat{\MMD}^2(p,\Pi)|$ statistic, $n=100$, $|\Pi|=$100,000}
    \label{fig:mmd_n100}
\end{minipage}
\begin{minipage}[b]{.49\linewidth}
    \centering
    \includegraphics[width=1\columnwidth]{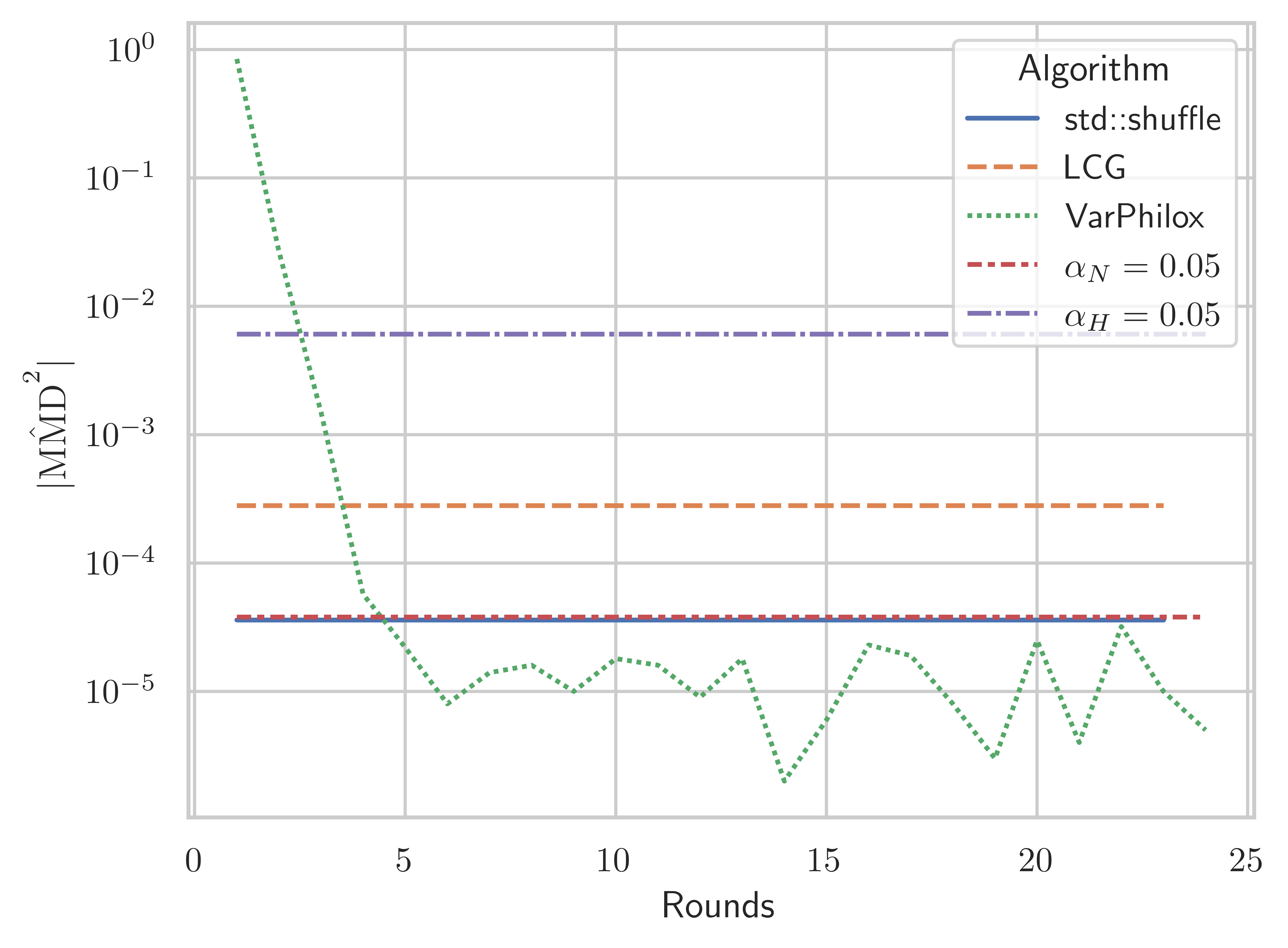}
    \caption{$|\hat{\MMD}^2(p,\Pi)|$ statistic, $n=1000$, $|\Pi|=$100,000}
    \label{fig:mmd_n1000}
\end{minipage}
\end{figure}

\section{Evaluation of throughput}
\label{sec:evaluation}
We now evaluate the throughput of our bijective shuffling method, where throughput refers to (millions of keys)/time(s). Unless otherwise stated, experiments are performed on an array of 64-bit keys of length $2^w+1$, where $w$ ranges from 8 to 29. This represents the worst-case scenario, where our bijective shuffle algorithm must redundantly evaluate $2^{w}-1$ elements. For all throughput results, we report the average of five trials.

To consider the effect of code optimisations, we evaluate three CUDA implementations of Algorithm \ref{alg:shuffle} with varying levels of GPU kernel fusion:

\begin{itemize}
  \item \textit{Bijective0}: The transformation ($b=f_n(i)$), stream compaction, and gather are implemented in separate passes.
  \item \textit{Bijective1}: The transformation, stream compaction, and gather are fused into a single scan operation. The two-pass scan algorithm of \cite{harris2007parallel} is used.
  \item \textit{Bijective2}: The transformation, stream compaction, and gather are fused into a single scan operation. The single pass scan algorithm of \cite{merrill2016single} is used, for $m$ total global memory reads/writes. 
\end{itemize}

Figure \ref{fig:fusion} plots the throughput of these variants, on a Tesla V100-32GB GPU, using VariablePhilox as the bijective function. For reference, the line labeled ``gather'' shows an upper bound on throughput for $n$ random gather operations in global memory. The optimised algorithm \textit{Bijective2} closely matches the optimal throughput of random gather, and is equivalent in performance for sizes $> 2^{21}$. \textit{Bijective2(n=m)} indicates the best-case performance of the algorithm, where the sequence is exactly a power of two length. The best-case and worst-case performance do not differ greatly because redundant elements do not incur global memory transactions, and are therefore relatively inexpensive to evaluate. The performance of random gather peaks around $2^{20}$, where there is sufficient L2 cache to mitigate the effects of uncoalesced reads/writes. At this size, there is a gap between \textit{Bijective2} and random gather due to arithmetic operations required in evaluating multiple rounds of the VariablePhilox function, and because the prefix-sum must be redundantly evaluated up to the nearest power of two (although no global memory transactions occur for these redundant elements). This gap disappears at larger sizes when the runtime of the kernel becomes dominated by memory operations.

\begin{figure}
    \centering
    \includegraphics[width=0.75\columnwidth]{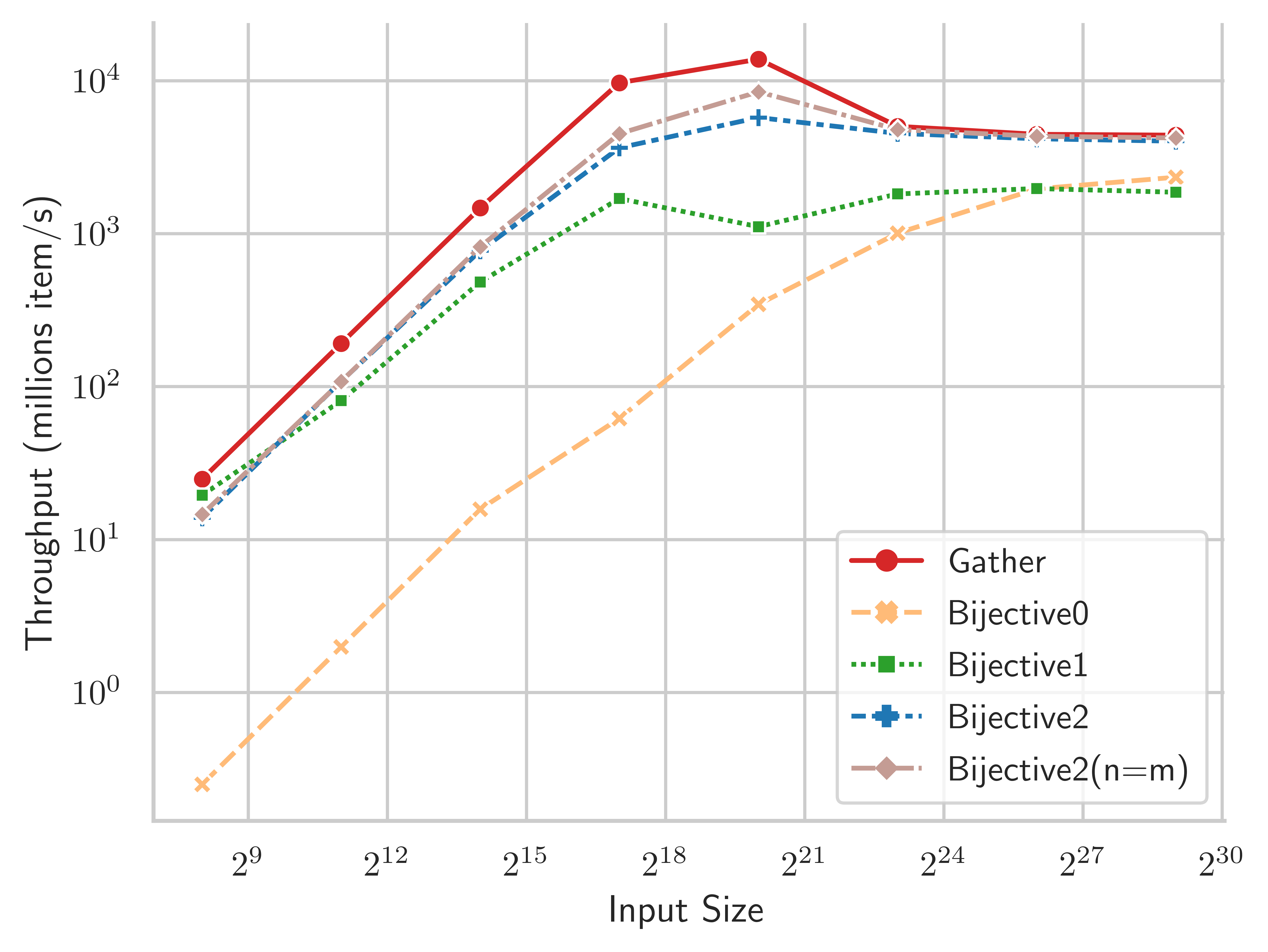}
    \caption{GPU kernel fusion for bijective shuffle --- increasing kernel fusion decreases memory transactions and approaches the gather throughput upper bound.}
    \label{fig:fusion}
\end{figure}

We now proceed to a comparison of different shuffling algorithms. Table \ref{tab:gpu_shuffling} and Figure \ref{fig:gpu_shuffling} show the throughput in millions of items per second, for the Tesla V100-32GB GPU, and Table \ref{tab:gpu_shuffling_2080} and Figure \ref{fig:gpu_shuffling_2080} show data for the GeForce RTX 2080 GPU. LCG and VariablePhilox implement optimised bijective shuffling using the methods described in Section \ref{sec:bijective}. The dart-throwing algorithm uses CUDA atomic exchange instructions to attempt to place items randomly in a buffer of size $2n$, a stream compaction operation is then applied to the buffer to provide the final shuffled output. SortShuffle applies the state-of-the-art radix sort algorithm of \cite{merrill2015cub} to randomly generated 64-bit keys. As discussed in Section \ref{sec:existing}, many existing work on shuffling reduce to sorting algorithms on infinite length keys. Thus, SortShuffle is representative of a wide class of divide-and-conquer algorithms.

The bijective shuffle algorithms with the VariablePhilox and LCG functions achieve near-optimal throughput at large sizes, where performance is dominated by global memory operations. Throughput for these methods is more than an order of magnitude higher than the DartThrowing or SortShuffle methods. The lesser throughput of DartThrowing can be explained by the overhead of generalised atomic instructions in the CUDA architecture, as well as contention among threads when multiple threads attempt to write to the same memory location. SortShuffle relies on radix sort, where prefix sum is applied to 4 bits in each pass, requiring 16 passes over the data to fully sort 64-bit keys. In comparison, bijective shuffle is designed to require only a single prefix sum operation. Bijective shuffle is also fully deterministic, unlike DartThrowing, and requires no global memory for working space, whereas DartThrowing and SortShuffle use memory proportional to $O(n)$.
\begin{table}
\caption{GPU shuffling throughput (millions items/s) - Tesla V100}
\label{tab:gpu_shuffling}
\centering
\begin{tabular}{rrrrrrrrr}
\toprule
{} & Gather & VarPhilox &    LCG & DartThrowing & SortShuffle \\
Input size   &        &        &        &              &             \\
\midrule
$2^{8} + 1$  &  24.64 &  14.61 &  16.21 &        2.459 &       1.845 \\
$2^{11} + 1$ &  188.8 &  115.4 &  124.8 &        15.83 &       8.191 \\
$2^{14} + 1$ &   1471 &  810.6 &  851.4 &         97.3 &       26.17 \\
$2^{17} + 1$ &   9696 &   3836 &   3968 &         145. &       69.26 \\
$2^{20} + 1$ &  13830 &   5800 &   5642 &        159.8 &       127.3 \\
$2^{23} + 1$ &   5036 &   4527 &   4476 &        150.9 &       132.6 \\
$2^{26} + 1$ &   4465 &   4172 &   4143 &        133.8 &       118.4 \\
$2^{29} + 1$ &   4409 &   4018 &   4011 &        123.2 &       111.8 \\
\bottomrule
\end{tabular}
\end{table}

\begin{table}
\caption{GPU shuffling throughput (millions items/s) - GeForce RTX 2080}
\label{tab:gpu_shuffling_2080}
\centering
\begin{tabular}{rrrrrrrr}
\toprule
{} &   Gather &  VarPhilox &     LCG &  DartThrowing &  SortShuffle \\
Input size   &          &            &         &               &              \\
\midrule
$2^{8} + 1$  &    35.63 &      13.73 &   14.19 &          2.98 &         2.16 \\
$2^{11} + 1$ &   273.70 &     109.07 &  112.40 &         18.10 &         8.54 \\
$2^{14} + 1$ &  2183.41 &     776.85 &  846.42 &         66.74 &        30.62 \\
$2^{17} + 1$ & 12225.35 &    3649.39 & 3759.00 &         95.38 &        71.17 \\
$2^{20} + 1$ &  4187.73 &    3641.61 & 3437.88 &         89.55 &        80.63 \\
$2^{23} + 1$ &  2238.91 &    2232.61 & 2260.85 &         80.67 &        75.03 \\
$2^{26} + 1$ &  2097.25 &    2096.47 & 2105.51 &         71.44 &        67.35 \\
\bottomrule
\end{tabular}
\end{table}

\begin{figure}
\centering
\begin{minipage}[b]{.49\linewidth}
    \centering
    \includegraphics[width=1\columnwidth]{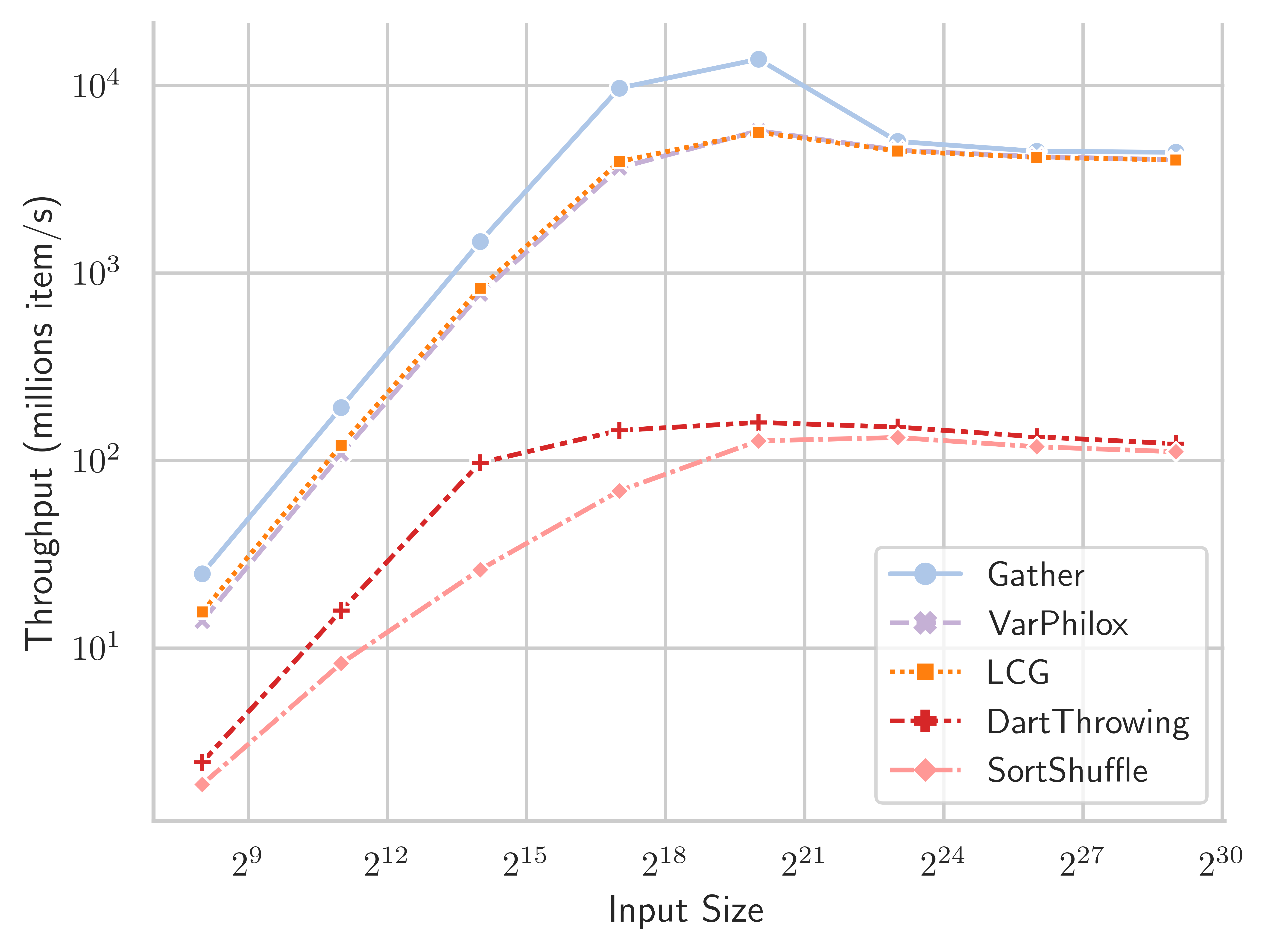}
    \caption{GPU shuffling algorithms - Tesla V100}
    \label{fig:gpu_shuffling}
\end{minipage}
\begin{minipage}[b]{.49\linewidth}
    \centering
    \includegraphics[width=1\columnwidth]{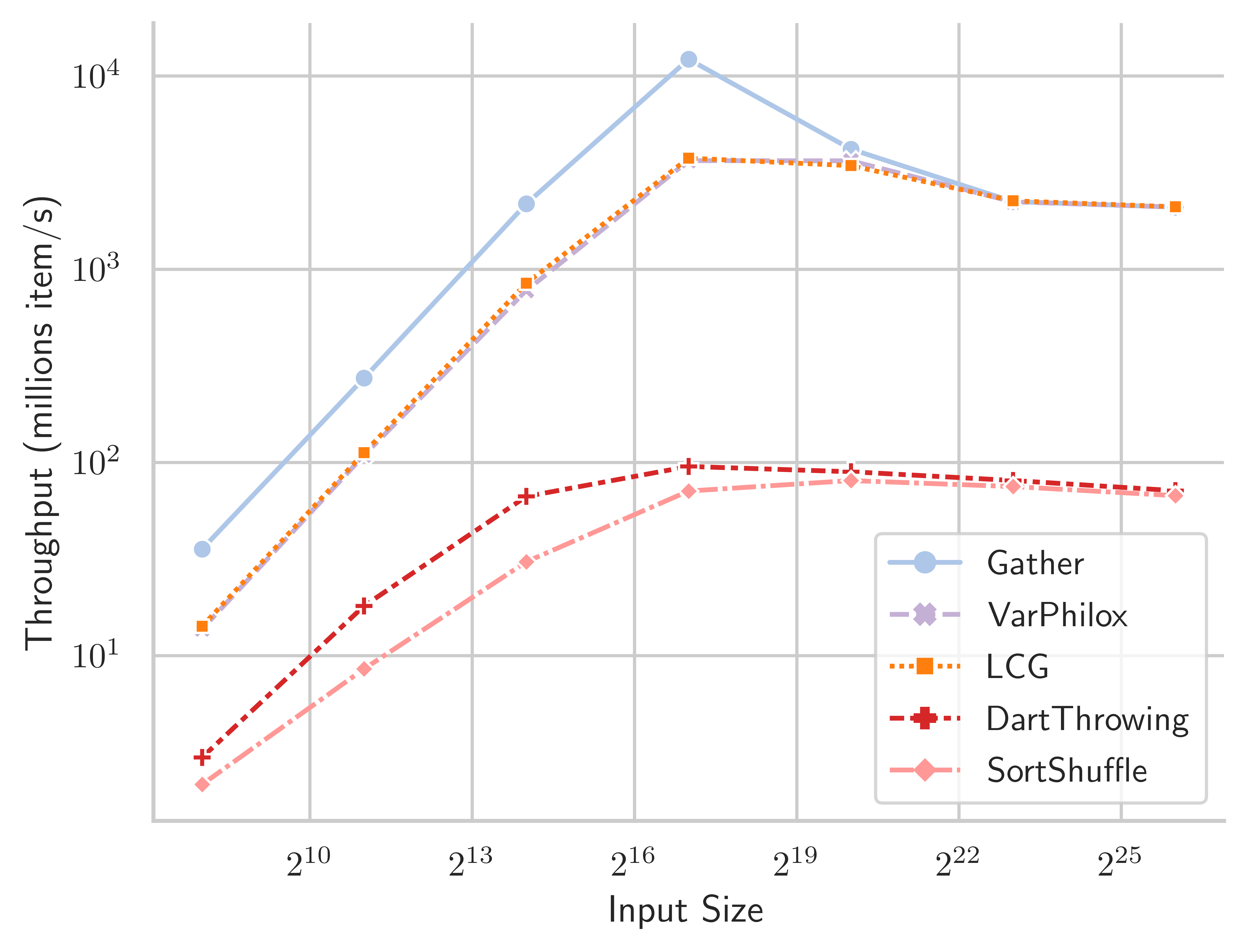}
    \caption{GPU shuffling algorithms - GeForce RTX 2080}
    \label{fig:gpu_shuffling_2080}
\end{minipage}
\end{figure}

We also evaluate the performance of our shuffling method using 2x Intel Xeon E5-2698 CPUs, with a total of 40 physical cores. Results are presented in Table \ref{tab:cpu_shuffling} and Figure \ref{fig:cpu_shuffling}. The bijective shuffle method is not designed for CPU architectures, but the results are informative nonetheless. The Gather, VariablePhilox (bijective shuffle) and SortShuffle algorithms are implemented using the Intel Thread Building Blocks (TBB) library~\cite{tbb}. The Rao-Sandelius (RS)~\cite{rao} and MergeShuffle~\cite{mergeshuffle} algorithms use the implementation of \cite{mergeshuffle}, and std::shuffle is the single-threaded C++ standard library implementation. VariablePhilox has the highest throughput for medium-sized inputs, although it is outperformed by RS when $n>2^{24}$. Curiously, RS has higher throughput at large sizes than the TBB gather implementation. This could indicate that RS has a more cache-friendly implementation, despite having a worse computational complexity of $O(n \log n)$, or simply a suboptimal implementation in TBB. Further work is needed to properly investigate these effects for CPU architectures.

\begin{table}
\caption{CPU shuffling throughput (millions items/s) - 2x Xeon E5-2698}
\label{tab:cpu_shuffling}
\centering
\begin{tabular}{rrrrrrrrr}
\toprule
{} &  Gather &  VarPhilox &  DartThrowing &  std::shuffle &     RS &  MergeShuffle &  SortShuffle \\
Input size   &         &         &               &               &        &               &              \\
\midrule
$2^{8} + 1$  &  128.03 &    2.81 &          0.09 &         40.77 &  41.18 &         46.24 &         2.45 \\
$2^{11} + 1$ &  635.04 &   17.65 &          0.69 &         63.83 &  63.82 &         64.95 &         6.26 \\
$2^{14} + 1$ &  935.55 &   89.88 &          4.76 &         68.05 &  68.00 &         67.85 &         8.30 \\
$2^{17} + 1$ &  434.46 &  207.94 &         21.69 &         63.05 &  78.71 &         52.03 &        18.60 \\
$2^{20} + 1$ &  353.72 &  246.79 &         39.53 &         59.38 &  99.76 &         32.94 &        79.30 \\
$2^{23} + 1$ &  108.23 &  157.88 &         21.37 &         45.67 & 108.78 &         23.44 &        60.11 \\
$2^{26} + 1$ &   64.32 &   49.02 &         16.04 &         24.93 & 111.15 &         18.38 &        47.45 \\
$2^{29} + 1$ &   54.12 &   44.06 &         15.63 &         21.66 & 104.99 &         15.03 &        39.23 \\
\bottomrule
\end{tabular}
\end{table}

\begin{figure}
    \centering
    \includegraphics[width=0.75\columnwidth]{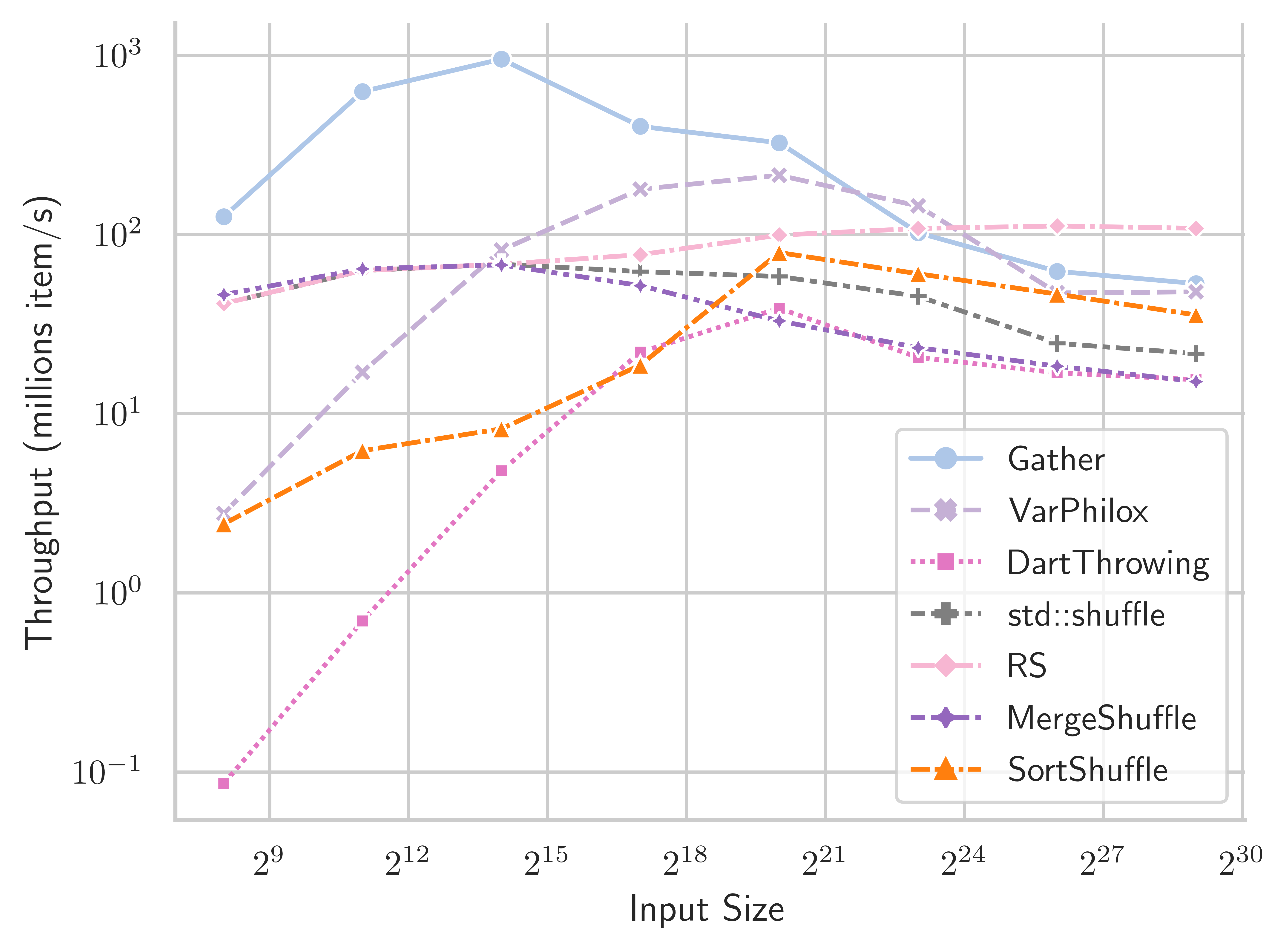}
    \caption{CPU shuffling algorithms - 2x Intel Xeon E5-2698}
    \label{fig:cpu_shuffling}
\end{figure}

\section{Limitations and Future Research}
While the presented algorithm for GPUs is highly effective in terms of throughput and the quality of its pseudorandom outputs, its parameterisation differs from that of typical standard library implementations of shuffle. For example, the C++ standard library implementation accepts any `uniform random bit generator'~\cite{stroustrup2013c++}, such as the common Mersenne-twister generator~\cite{mersenne}. Our implementation is instead parameterised by the selection of a bijective function, and a key of length $\log(n)k$ bits, where $n$ is the length of the sequence to be shuffled, and $k$ is the number of rounds in the cipher.

Future work may explore other useful constructions of bijective functions, beyond the LCG and Feistel variants discussed in this paper.

It also should be noted that the Fisher-Yates and Rao-Sandelius algorithms may operate in-place, and our proposed GPU algorithm does not, requiring the allocation of an output buffer. This is not unexpected, as there are few truly in-place GPU algorithms that reorder inputs (\cite{peters2009fast} is a notable exception). For future work, it would be interesting to consider if an in-place parallel shuffling algorithm is possible for GPU architectures.

\section{Conclusion}
We provide an algorithm for random shuffling specifically optimised for GPU architectures, using modified bijective functions from cryptography. Our algorithm is highly practical, achieving performance approaching the maximum possible device throughput, while being fully deterministic, using no extra working space, and providing high-quality distributions of random permutations. An outcome of this work is also a statistical test for uniform distributions of permutations based on the Mallows kernel that we expect to be useful beyond shuffling algorithms.

\bibliographystyle{ACM-Reference-Format}
\bibliography{manuscript}

\appendix
\section{Runtime of GPU and CPU shuffling algorithms}
\label{app:runtime}
Figures \ref{fig:gpu_shuffling_runtime} and \ref{fig:cpu_shuffling_runtime} reproduce Figures \ref{fig:gpu_shuffling} and \ref{fig:cpu_shuffling}, reporting runtime in seconds instead of throughput ((millions of keys)/time(s)). 
\begin{figure}
\centering
\begin{minipage}[b]{.49\linewidth}
    \centering
    \includegraphics[width=1\columnwidth]{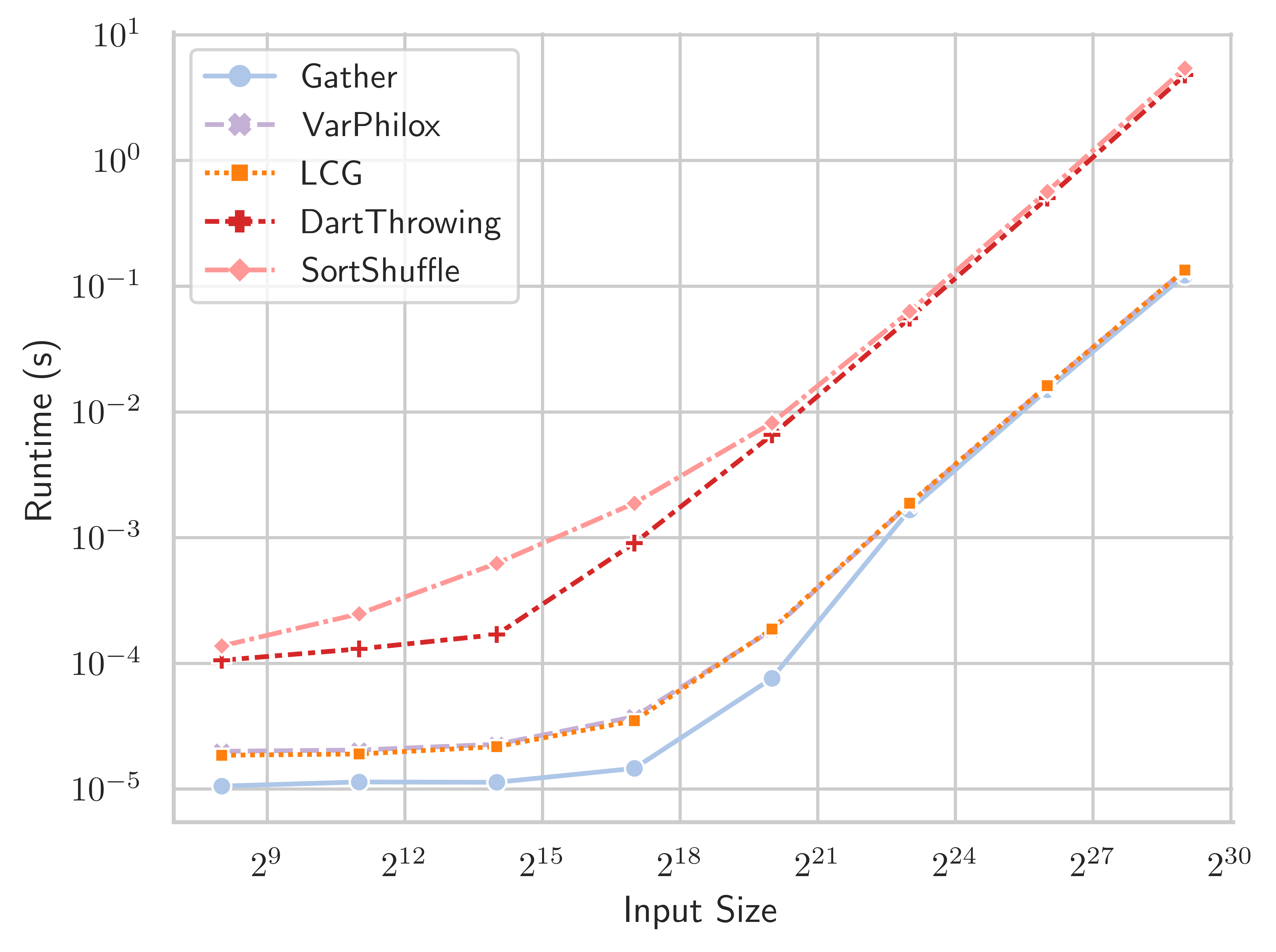}
    \caption{GPU shuffling algorithms runtime - Tesla V100}
    \label{fig:gpu_shuffling_runtime}
\end{minipage}
\begin{minipage}[b]{.49\linewidth}
    \centering
    \includegraphics[width=1\columnwidth]{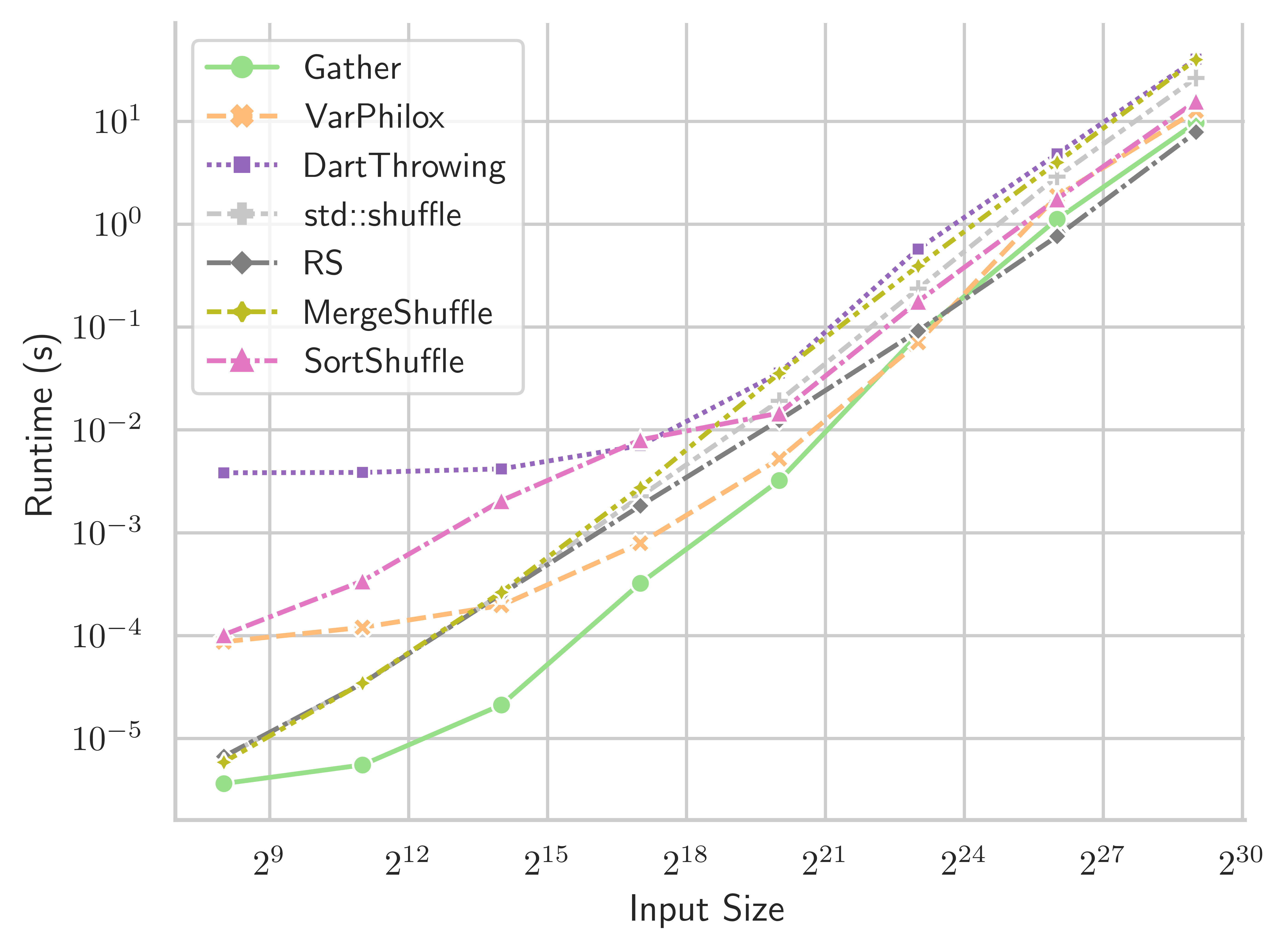}
    \caption{CPU shuffling algorithms runtime - 2x Intel Xeon E5-2698}
    \label{fig:cpu_shuffling_runtime}
\end{minipage}
\end{figure}

\end{document}